\documentclass[runningheads]{llncs}

\usepackage{amsmath,amsfonts}
\usepackage{listings}
\usepackage{hyperref}
\usepackage[capitalize,nameinlink]{cleveref}
\usepackage{proof-dashed}
\usepackage{multicol}

\usepackage[T1]{fontenc}
\usepackage{graphicx}

\newcommand*{\m}[1]{\ensuremath{\mathsf{#1}}}
\newcommand*{\mi}[1]{\ensuremath{\mathit{#1}}}
\newcommand*{\rname}[1]{\textsc{#1}}

\newcommand*{\com}{\ensuremath{\, \mbox{\tt{-\!\raisebox{-0.12ex}{>}}}\,}}
\newcommand*{\tricom}{\ensuremath{\, \mbox{\tt{\raisebox{-0.12ex}{<}\!-\!\raisebox{-0.12ex}{>}}}\,}}
\newcommand*{\reduce}[1]{\xrightarrow[]{#1}}

\newcommand*{\nil}{\textbf{0}}

\begin{document}
\title{From Infinity to Choreographies}
\subtitle{Extraction for Unbounded Systems}
\titlerunning{From Infinity to Choreographies}
\author{Bjørn Angel Kjær
  \and
  Luís Cruz-Filipe
  \and
  Fabrizio Montesi
}

\authorrunning{B.A.~Kjær et al.}

\institute{
  Department of Mathematics and Computer Science,
  University of Southern Denmark\\
  Campusvej 55, 5230 Odense M, Denmark
  \email{bjoernak@gmail.com,\{lcfilipe,fmontesi\}@imada.sdu.dk}
}

\maketitle

\begin{abstract}
  Choreographies are formal descriptions of distributed systems, which focus on the way in which
  participants communicate.
  While they are useful for analysing protocols, in practice systems are written directly
  by specifying each participant's behaviour.
  This created the need for \emph{choreography extraction}: the process of obtaining a choreography
  that faithfully describes the collective behaviour of all participants in a distributed protocol.

  Previous works have addressed this problem for systems with a predefined, finite number of
  participants.
  In this work, we show how to extract choreographies from system descriptions where the total
  number of participants is unknown and unbounded, due to the ability of spawning new processes at
  runtime.
  This extension is challenging, since previous algorithms relied heavily on the set of possible
  states of the network during execution being finite.

\end{abstract}

\section{Introduction}
\emph{Choreographies} are coordination plans for concurrent and distributed systems, which describe the expected interactions that system participants should enact~\cite{CM20,M22}. Languages for expressing choreographies (choreographic languages) are widely used for documentation and specification purposes, some notable examples being Message Sequence Charts~\cite{msc}, UML Sequence Diagrams~\cite{uml}, and choreographies in the Business Process Modelling Notation (BPMN)~\cite{bpmn}.
More recently, such languages have also been used for programming and verification, e.g., as in choreographic programming~\cite{M13:phd} and multiparty session types~\cite{HYC16} respectively.

In practice, many system implementations do not come with a choreography yet.
\emph{Choreography extraction} (extraction for short) is the synthesis of a choreography that faithfully represents the specification of a system based on message passing (if it exists)~\cite{CMS18,CLMS22,LT12,LTY15}. 
Extraction is helpful because it gives developers a global overview of how the \emph{processes} (abstractions of endpoints) of a system interact, making it easier to check that they collaborate as intended.
In general, though, it is undecidable whether such a choreography exists, making extraction a challenging problem.

Current methods for extraction cannot analyse systems that spawn new processes at runtime: they can only deal with systems where the number of participants is finite and statically known.
This is an important limitation, because many modern distributed systems dynamically create
processes for several reasons, such as scalability.
The aim of this article is to address this shortcoming.

As an example of a system that cannot be analysed by previous work, consider a simple implementation of a serverless architecture (\cref{example:serverlesspseudo}).
\begin{example}
A simple serverless architecture, written in pseudocode (we will formalise this example later in the article). A client sends a request to an entry-point, which then spawns a temporary process to handle the client. The spawned process computes the response to the request. Then it offers the client to make another request, in which case a new process is spawned to handle that, since the new request may require a different service.
\begin{lstlisting}[escapeinside=**, keywords={send, loop, introduces, receive, if, then, else, terminate, request}]
*\textsf{\textbf{client:}}*
	send init to server;
    loop {
        server presents worker; receive result from worker; 
        if finished
           then request termination from worker; terminate;
           else request next from worker; server:=worker;
    }

*\textsf{\textbf{server:}}*
	receive init from client; Handle(server);

	procedure Handle(parent) {
		spawn worker with {
			parent presents client; ComputeResult();
			send result to client; receive request from client;
			switch request {
				next: Handle(worker);
				termination: terminate;
			}
		}
		introduce worker and client; terminate;
	}
\end{lstlisting}
\label{example:serverlesspseudo}
\end{example}
\cref{example:serverlesspseudo} illustrates the usefulness of extraction: a human could manually go through the code and check that the two processes will interact as intended; however, despite it being a greatly simplified example, it is not immediately obvious whether the processes will communicate correctly or not.

Previous methods for extraction use graphs to represent the possible (symbolic) execution space of
the system under consideration. The challenge presented by code as in
\cref{example:serverlesspseudo} is that these graphs are not guaranteed to be finite anymore,
because of the possibility of spawning new processes at runtime.

In this article we introduce the first method for extracting choreographies that supports process spawning, i.e., the capability of creating new processes at runtime.
Our main contribution consists of a theory and implementation of extraction that use name substitutions to obtain finite representations of infinite symbolic execution graphs.
Systems with process spawning have a dynamic topology, which further complicates extraction: new processes can appear at runtime, and can then be connected to other processes to enable communication. We extend the languages used for extraction in previous work with primitives for capturing these features, and show that our method can deal with them.

\paragraph{Structure of the paper.}
In \cref{section:oldalgorithm}, we recap the basic theory of extraction.
In \cref{sec:extraction}, we introduce the languages for representing systems (as networks of processes) and choreographies.
\Cref{section:processspawning} reports our method for extracting networks with an unbounded number of
processes (due to spawning), its implementation and limitations.
We conclude in \cref{sec:concl}.

\paragraph{Related work.}
We have already mentioned most of the relevant related work in this section. Choreography extraction has been explored for languages that include internal computation~\cite{CLMS22}, process terms that correspond to proofs in linear logic~\cite{CMS18}, and session types (abstract terms without internal computation)~\cite{LT12,LTY15}. Our method deals with the first case (the most general among those cited). Our primitives for modelling process spawning are inspired by~\cite{CM17}.

\section{Background}
\label{section:oldalgorithm}
This section summarizes the framework for choreography extraction that we
extend~\cite{CLM17,CLMS22,injava,SafinaPhD}.
The remainder of the article expands upon this work to extend the capabilities of extraction, and to
bring it closer to real systems.

\subsection{Networks}
Distributed systems are modelled as \emph{networks}, which consist of several participants executing
in parallel.
Each participant is called a \emph{process}, and the sequence of actions it executes is called a
\emph{behaviour}.
Each process also includes a set of \emph{procedures}, consisting of a name and associated
behaviour.

Behaviours are formally defined by the grammar given below.
\begin{align*}
  B ::= & \ \nil
  \mid X
  \mid \m{p}!m; B
  \mid \m{p}?; B
  \mid \m{p}\oplus \ell; B
  \mid \m{p}\& \{\ell_1:B_1,\ldots,\ell_n:B_n\} \\
  & \mid \texttt{if }\! e \!\texttt{ then }\! B_1 \texttt{ else }\! B_2
\end{align*}
Term $\nil$ designates a terminated process.
Term $X$ invokes the procedure named $\textup{X}$.
Invoking a procedure must be the last action on a behaviour -- in other words, we only allow tail
recursion.
Term $\m{p}!m; B$ describes a behaviour where the executing process evaluates message $m$ and sends
the corresponding value to process $\m{p}$, then continues as $B$.
The dual term $\m{p}?; B$ describes receiving a message from $\m{p}$, storing it locally and
continuing as $B$.
Behaviour $\m{p}\oplus\ell; B$ sends the selection of label $\ell$ to process $\m{p}$ then continues as
$B$, while the dual $\m{p}\& \{\ell_1:B_1,\ldots,\ell_n:B_n\}$ offers the behaviours
$B_1,\ldots,B_n$ to $\m{p}$, which can be selected by the corresponding labels
$\ell_1,\ldots,\ell_n$.
Finally, $\texttt{if } e \texttt{ then } B_1 \texttt{ else } B_2$ is the conditional term: if the
Boolean expression $e$ evaluates to true, it continues as $B_1$, otherwise, it continues as $B_2$.

The syntax
\begin{lstlisting}[mathescape=true]
$\m{p}$ {
    def $X_1$ { $B_1$ }
    $\vdots$
    def $X_n$ { $B_n$ }
    main { $B$ }
}
\end{lstlisting}
describes a process named $\m{p}$ with local procedures $X_1,\ldots,X_n$ defined respectively as
$B_1,\ldots,B_n$, intending to execute behaviour $B$.

A network is specified as a sequence of processes, all with distinct names, separated by
vertical bars (\lstinline+|+).
\cref{example:onlinestorenetwork} defines a valid network, which we use as running example
throughout this section to explain the existing extraction algorithm.

\begin{example}
  This network describes the protocol for an online store.
  The customer sends in items to purchase, then asks the store to proceed to checkout, or continue
  browsing.

  Once the customer proceeds to checkout, they send their payment information to the store.
  The store then verifies that information, and either completes the transaction, or asks the client
  to re-send payment information if there where a problem.
\begin{lstlisting}[mathescape=true, escapeinside=**]
$\m{customer}$ {
	def *\mi{browse}*{ *\m{store}*!*\mi{item}*; if *\mi{checkout}*
	    then *\m{store}*$\oplus\mi{buy}$; *\mi{purchase}*;
	    else *\m{store}*$\oplus\mi{more}$; *\mi{browse}* }
	def *\mi{purchase}*{ *\m{store}*!*\mi{payment}*; *\m{store}*&{*\mi{accept}*: *\nil*, *\mi{reject}*: *\mi{purchase}*} }
	main{ *\mi{browse}* } 
} |
$\m{store}$ {
	def *\mi{offer}*{ *\m{customer}*?; *\m{customer}*&{*\mi{buy}*: *\mi{payment}*, *\mi{more}*: *\mi{offer}*} } 
	def *\mi{payment}*{ *\m{customer}*?; if *\mi{accepted}*
	    then *\m{customer}*$\oplus\mi{accept}$; *\nil*
	    else *\m{customer}*$\oplus\mi{reject}$; *\mi{payment}* } 
	main { *\mi{offer}* }
}
\end{lstlisting}
\label{example:onlinestorenetwork}
\end{example}

\subsection{Choreographies}
Global descriptions of distributed systems, specifying interactions between participants rather than their individual actions, are called \emph{choreographies}.
Similar to processes in networks, a choreography contains a set of procedure definitions, and a main body.
The terms of choreography bodies are defined by the grammar below and closely correspond to the actions in process behaviours.
\begin{align*}
  C ::= & \ \nil
  \mid X
  \mid \m{p}.m \com \m{q}; C
  \mid \m{p} \com \m{q}[\ell]; C
  \mid \texttt{if } \m{p}.e \texttt{ then } C_1 \texttt{ else } C_2
\end{align*}
Term $\nil$ denotes a choreography body where all processes are terminated.
Term $X$ invokes the procedure with name $\textup{X}$.
In the communication $\m{p}.m \com \m{q}; C$, process \m{p} sends message $m$ to \m{q}, which stores
the result, and the system continues as described by choreography body $C$.
Likewise, in label selection $\m{p} \com \m{q}[\ell]; C$, process \m{p} selects an action in \m{q}
by sending the label $\ell$, and the system continues as $C$.
In the conditional $\texttt{if } \m{p}.e \texttt{ then } C_1 \texttt{ else } C_2$, process \m{p} starts by evaluating the Boolean expression $e$; if this resolves to true, then the choreography continues as $C_1$, otherwise it continues as $C_2$. 

\begin{example}
The protocol described by the network in \cref{example:onlinestorenetwork} can be written as the following choreography.
\begin{lstlisting}[mathescape=true, escapeinside=**]
def *\mi{Buy}* { 
    *\m{customer}*.$item$*\com\m{store}*; if *\m{customer}*.$checkout$
        then *\m{customer}\com\m{store}*$[buy]$; *\mi{Pay}*
        else *\m{customer}\com\m{store}*$[more]$; *\mi{Buy}*
} 
def *\mi{Pay}* { 
    *\m{customer}*.*\mi{payment}\com\m{store}*; if *\m{store}*.*\mi{accepted}*
        then *\m{store}\com\m{customer}*$[accept]$; *\nil* 
        else *\m{store}\com\m{customer}*$[reject]$; *\mi{Pay}*
} 
main {*\mi{Buy}*}
\end{lstlisting}
\label{example:onlinestore-choreography}
\end{example}

\subsection{Extraction algorithm}
The extraction algorithm from~\cite{CLM17,CLMS22} consists of two steps.
The first step is building a graph that represents a symbolic execution of the network.
The second is to traverse this graph, using its edges to build the extracted choreography.

\subsubsection{Graph generation.}
The first step in extracting a choreography from a network is building a \emph{Symbolic Execution
  Graph} (SEG) from the network.
A SEG is a directed graph representing an abstraction of the possible evolutions of the network over
time.
It abstracts from the concrete semantics by ignoring the concrete values being communicated and
considering both possible outcomes for every conditional.
Nodes contain possible states of the network, and edges connect nodes that are related by execution
of one action (the label of the edge).\footnote{The formal details can be found in~\cite{CLM17,CLMS22}.}

Edges in the SEG are labeled by \emph{transition labels},
which represent the possible actions executable by the network: value communications (matching a
send action with the corresponding receive), label selection (matching selection and offer), and
conditionals.
For the last there are two labels, representing the two possible outcomes (the ``then'' and ``else''
branch, respectively).
  \[
  \lambda ::= \m{p}.e \com \m{q}
  \mid \m{p} \com \m{q}[\ell]
  \mid \m{p}.e \texttt{ then} \mid \m{p}.e \texttt{ else}
  \]

As an example, we show how to build the SEG for the network in \cref{example:onlinestorenetwork}
(see \cref{fig:seg3}).
The main behaviours of the two processes are procedure invocations (node on top).
Expanding the corresponding definitions, we find out that the first action by \m{customer} is
sending to \m{store}, while \m{store}'s first action is receiving from \m{customer}.
These actions match, so the network can execute an action reducing both \m{store} and \m{customer}.
This results in a new network, which is placed in a new node, and we connect both nodes by the
transition label describing the executed action.

The next action by \m{customer} is receiving a label from \m{store}, but \m{store} needs to evaluate
a conditional expression to decide which label to send.
There are two possible outcomes for these evaluation, so we create two new nodes and label the edges towards them with
the corresponding possibilities (\texttt{then} or \texttt{else}).
Continuing to expand the \texttt{else} branch leads to a network that is already in the SEG, so we simply
add an edge to the node containing that network.
The \texttt{then} branch evolves in two steps into a second conditional, whose \texttt{else} branch again creates
a loop, while its \texttt{then} branch evolves into a network where all processes has terminated.
This concludes the construction of the SEG.

\begin{figure}[!htb]
\includegraphics[width=\textwidth]{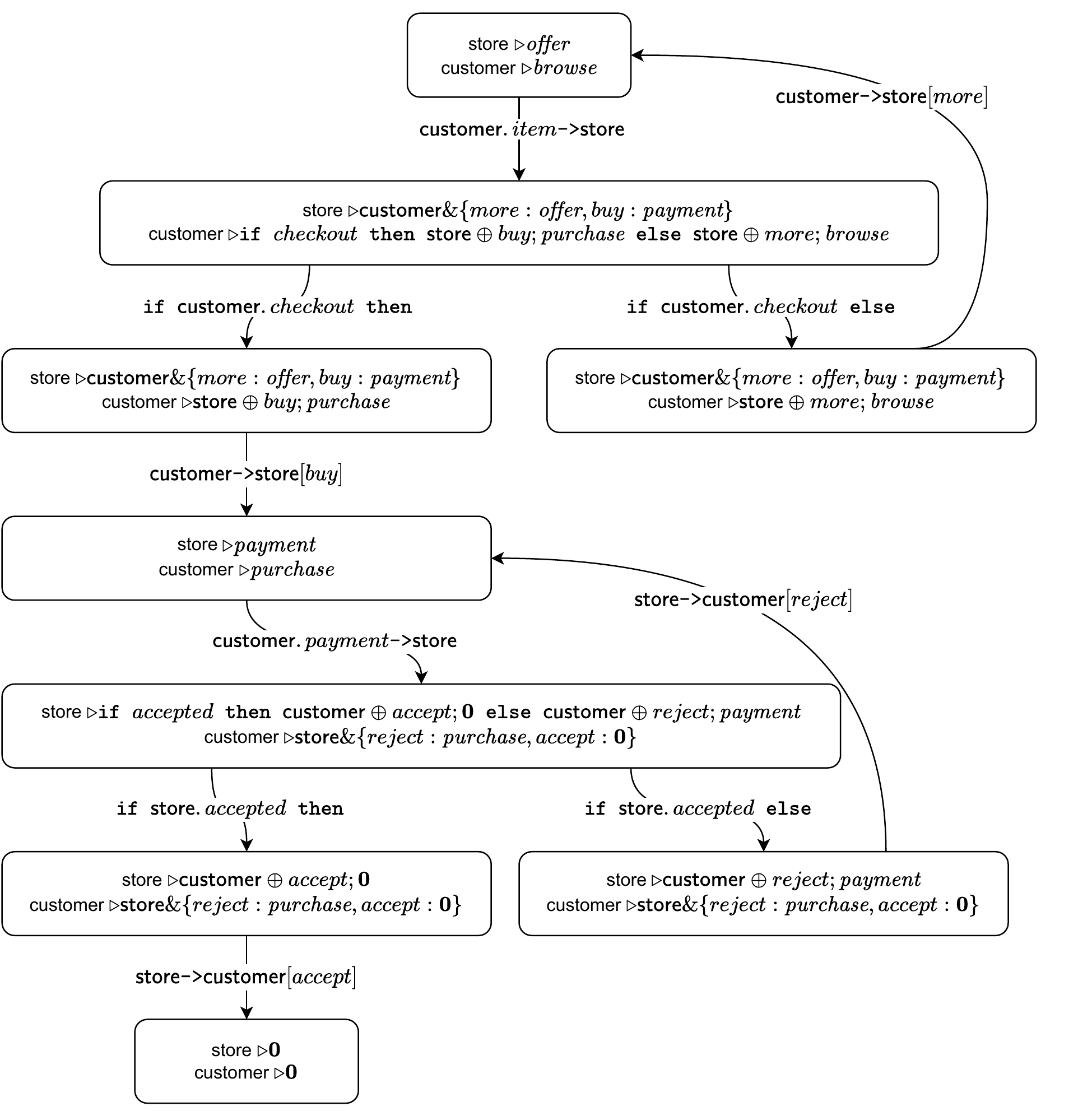}
\caption{The SEG of the network from \cref{example:onlinestorenetwork}.}
\label{fig:seg3}
\end{figure}

In this example SEG generation went flawlessly, but that is not always the case.
We saw that a process trying to send a message must wait for the receiving process to be able to
execute a matching receive; this can lead to situations where the network is \emph{deadlocked} -- no
terms can be executed.
In that case, the behaviour of the network cannot be described by a choreography, and the network
cannot be extracted.

This algorithm also relies on the fact that network execution is confluent: the success of
extraction does not depend on which action is chosen when constructing the SEG, in case several are
possible.
(This can affect the algorithm's performance, though.)
Furthermore, guaranteeing that all possible evolutions of the network are captured requires some
care when closing loops: all processes must reduce in every cycle in the SEG.
This is achieved by marking processes in the network and checking that every loop contains a node
where all processes are marked.
These aspects are orthogonal to the current development, and we refer the interested reader
to~\cite{CLMS22} for details.

\subsubsection{Choreography construction.}
\label{sec:prevchorgen}
The main idea for generating a choreography from a SEG is that edges correspond to choreography
actions, so the choreography essentially describes all paths in the SEG.
The choreographic way of representing loops is by means of procedures, so each loop in the SEG
should become a procedure definition.
To achieve this, we first \emph{unroll} the graph by splitting every \emph{loop node} -- the nodes
that close a loop\footnote{Formally, every node with at least two incoming edges -- plus the
  starting node, if it has any incoming edges} -- into two: a \emph{exit node}, which is the target
of all edges previously pointing to the loop node, and a \emph{entry node}, which is the source of
all edges previously pointing from the loop node.
Entry nodes are given distinct procedure names, and exit nodes are associated with the corresponding
procedure calls.
The unrolled SEG is now a forest with each tree representing a procedure, as shown in
\cref{fig:segunrolled}.
\begin{figure}[!htb]
\includegraphics[width=\textwidth]{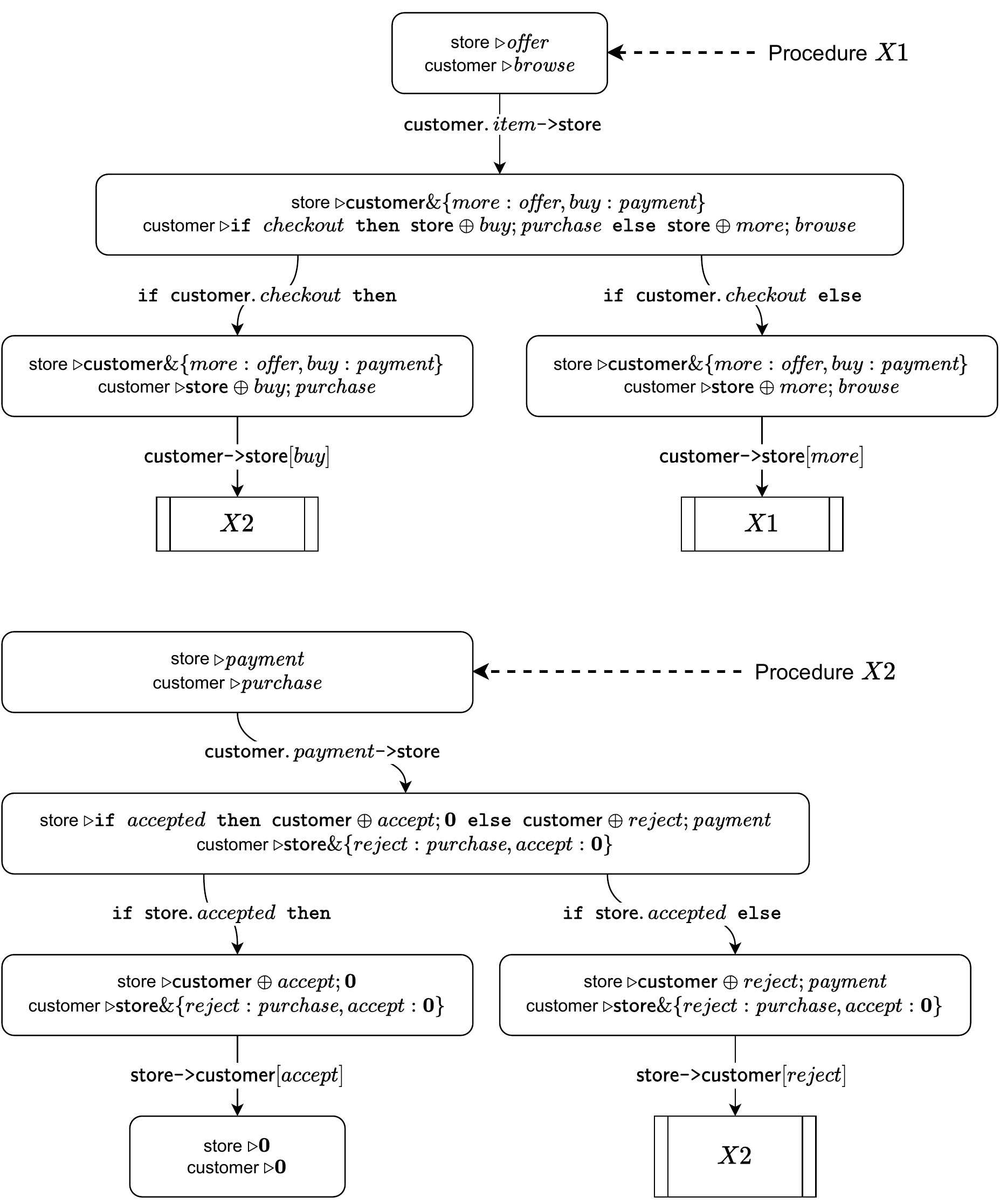}
\caption{The unrolling of the graph of \cref{fig:seg3} resulting in two trees, each corresponding to
  a procedure in the SEG to be extracted. For readibility, the exit node corresponding to the
  invocation of $X2$ is depicted twice.}
\label{fig:segunrolled}
\end{figure}

Since transition labels are similar to the choreography body terms, it is simple to read
choreography bodies directly from each tree of the unrolled graph.
This is done recursively, starting from the root of each tree and proceeding as follows: when
encountering a node with no outgoing edges, then either all processes have terminated, in which case
we return $\nil$, or the node is an exit node, in which case we return the corresponding procedure
invocation.
If there is one outgoing edge, that edge represents an interaction, so we return the choreography
body that starts with the transition label for that edge and continues as the result of the
recursive invocation on the edge's target.
If there are two outgoing edges, then we return a conditional choreography body whose two
continuations are the results of the recursive calls targets of the two edges, as dictated by those
edges' labels.

\begin{example}
  By reading the trees in \cref{fig:segunrolled} in the manner described, we obtain the choreography
  given earlier in \cref{example:onlinestore-choreography}, where procedures $Buy$ and $Pay$ are now
  called $X1$ and $X2$, respectively.
\end{example}

\section{Networks and Choreographies with Process Spawning}
\label{sec:extraction}
In this section, we extend the theories of networks and choreographies with support for spawning new processes at runtime.

We start by adding three primitives to the language of behaviours, following ideas from~\cite{CM17}.

\begin{enumerate}
\item\emph{Spawning of processes.} The language of behaviours is extended with the construct
  $\texttt{spawn } \m{q} \texttt{ with } B_{\m q} \texttt{ continue } B$,
  which adds a new process to the network with a new, unique name.
  The new process gets main behaviour $B_{\m q}$, and inherits its parent's set of procedures, while
  the parent continues executing $B$.
  This term also binds \m{q} (a process variable) in $B$.

\item\emph{Advertising processes.} Since names of newly spawned processes are only known to their
  parents, we need terms for communicating process names.
  Process $\m{p}$ can  ``introduce'' \m{q} and \m{r} to each other (send each of them the other's
  name) by executing term $\m{q} \tricom \m{r}; B$, while $\m{q}$ and $\m{r}$ execute the dual
  actions $\m{p}?\m{t};\; B_{\m q}$ and $\m{p}?\m{t};\; B_{\m r}$.
  Here \m{t} is again a variable, which is bound in the continuations $B_{\m q}$ and $B_{\m r}$.

\item\emph{Parameterised procedures.} To be able to use processes spawned at runtime in procedures,
  their syntax is changed so that they can take process names as parameters.
\end{enumerate}

We do not distinguish process names from process variables syntactically, as this simplifies the
semantics.
We assume as usual that all binders in the same term bind distinct variables, and work up to
$\alpha$-renaming.
However, we allow a variable to occur both free and bound in the same term -- this is essential for
our algorithm.

As previously, the semantics includes a state function $\sigma$, mapping each process to a value
(its memory state).
The new ingredient is a graph of connections $\mathcal G$ between processes, connecting pairs of
process that are allowed to communicate.
The choice of the initial graph allows for modelling different network topologies.
We use the notation $\m{p}\leftrightarrow\m{q}\in\mathcal G$ to denote that \m{p} and \m{q} are
connected in $\mathcal G$, and $\mathcal G\cup\{\m{p}\leftrightarrow\m{q}\}$ to denote the graph
obtained from $\mathcal G$ by adding an edge between \m{p} and \m{q}.

\cref{fig:networksemantics-main} includes some representative rules of this extended
semantics.\footnote{The complete semantics is given in \cref{section:fullproofs}.}
\begin{figure}[!ht]
\[
\infer[\rname{N|Com}]
      {\m{p} \triangleright \m{q}!e; B_1 \mid \m{q} \triangleright \m{p}?; B_2,\sigma,\mathcal G
        \reduce{\m{p}.v \com \m{q}}
        \m{p} \triangleright B_1 \mid \m{q} \triangleright B_2,\sigma[\m{q} \mapsto v],\mathcal G}
      {\m{p} \leftrightarrow \m{q} \in \mathcal G & e \downarrow^\sigma_\m{p} v}
\]

\[
\infer[\rname{N|Intro}]
      {\begin{array}{c}
          \m{p} \triangleright \m{q}\tricom \m{r}; B_\m{p} \mid \m{q} \triangleright \m{p}?\m{r}; B_\m{q} \mid \m{r} \triangleright \m{p}?\m{q}; B_\m{r},\sigma,\mathcal G \\ 
          \reduce{\m{p}.\m{q} \tricom \m{r}}\\
          \m{p} \triangleright B_\m{p} \mid \m{q} \triangleright B_\m{q} \mid \m{r} \triangleright B_\m{r}, \sigma, \mathcal G\cup\{\m{q}\leftrightarrow\m{r}\}
      \end{array}}
      {\m{p}\leftrightarrow\m{q}\in\mathcal G & \m{p}\leftrightarrow\m{r}\in\mathcal G}
\]

\vspace{\baselineskip}
\[
\infer[\rname{N|Spawn}]
      {\begin{array}{c}
          \m{p} \triangleright \texttt{spawn \m{q} with $B_q$ continue $B$},\sigma,\mathcal G \\
          \reduce{\m{p} \texttt{ spawns } \m{q}} \\
          \m{p} \triangleright B \mid \m{q} \triangleright B_q, \sigma, \mathcal G\cup\{\m{p}\leftrightarrow\m{q}\}
      \end{array}}
      {}
\]
\caption{New semantics of networks, selected rules. }
\label{fig:networksemantics-main}
\end{figure}

Rule \rname{N|Com} describes a communication.
Process \m{p} wants to send the result of evaluating $e$ to process \m{q}, and \m{q} is expecting to
receive from \m{p}.
These processes can communicate, and the result $v$ of evaluating $e$ at \m{p} is sent and stored in
\m{q} (premise $e \downarrow^\sigma_\m{p} v$).
The difference from the previous semantics is the presence of the additional premise
$\m{p}\leftrightarrow\m{q}\in\mathcal G$, which checks that these two processes are allowed to communicate.

Rule \rname{N|Intro} is similar, but process names are communicated instead, and the communication
graph is updated.
For simplicity, instead of explicitly substituting variables for process names, we assume that the
behaviours of \m{q} and \m{r} have previous been $\alpha$-renamed appropriately (this kind of simplifications based on $\alpha$-renaming are standard in process calculi~\cite{S95}).

Rule \rname{N|Spawn} creates a new process \m{q} into the network with a unique name, and adds an
edge between it and its parent to the network.
Note that \m{q} is distinct from other process names in the network.

Choreographies get two corresponding actions: $\texttt{\m{p} spawns \m{q}};\; C$ and
$\m{p.q} \tricom \m{r};\; C$.
Procedures also become parameterized.
At the choreography level we do not require process variables except in procedure definitions (which
are replaced by process names when called); we assume that all names of spawned processes are
unique (again treating \texttt{spawn} actions as binders and $\alpha$-renaming in procedure bodies
when needed at invocation time).

The corresponding rules for the semantics are given in \cref{fig:chorsemantics-main}.
\begin{figure}
\[
\infer[\rname{C|Com}]
      {\m{p}.e \com \m{q} ; C, \sigma, \mathcal G
        \reduce{\m{p}.v \com \m{q}}
        C,\sigma[\m{q} \mapsto v],\mathcal G}
      {\m{p}\leftrightarrow\m{q}\in\mathcal G & e \downarrow^\sigma_\m{p} v}
\]

\[
\infer[\rname{C|Intro}]
      {\m{p}.\m{q} \tricom \m{r} ; C, \sigma, \mathcal G
        \reduce{\m{p}.\m{q} \tricom \m{r}}
        C, \sigma, \mathcal G\cup\{\m{q}\leftrightarrow\m{r}\}}
      {\m{p}\leftrightarrow\m{q}\in\mathcal G & \m{p}\leftrightarrow\m{r}\in\mathcal G}
\]

\[
\infer[\rname{C|Spawn}]
      {\m{p} \texttt{ spawns } \m{q} ; C, \sigma, \mathcal G
        \reduce{\m{p} \texttt{ spawns } \m{q}}
        C,\sigma,\mathcal G\cup\{\m{p}\leftrightarrow\m{q}\}}
      {}
\]
\caption{New semantics of choreographies, selected rules.}
\label{fig:chorsemantics-main}
\end{figure}

\begin{example}
  We illustrate a network with process spawning by writing \cref{example:serverlesspseudo} in our
  language.
  The client sends a request to an entry-point, which then spawns an instance to handle the request,
  which gets introduced to the client.
  The instance sends the client the result, and the client either makes another request or
  terminates the connection.
  Additional requests spawn new instances, since requests may differ in kind -- this is handled by
  the previous instance to reduce load on the entry-point.

\begin{minipage}[t]{.45\textwidth}%
\begin{lstlisting}[mathescape=true, escapeinside=**]
$\m{client}$ {
    def $X$(s){
        *\m{s}*?*\m{w}*; *\m{w}*?; if *\mi{more}*
            then *\m{w}*$\oplus next$; $X$(*\m{w}*)
            else *\m{w}*$\oplus end$; $\nil$
    }
    main{ *\m{entry}*!$req$; $X$(*\m{entry}*) }
}
\end{lstlisting}
\end{minipage}
\begin{minipage}[t]{.45\textwidth}%
\begin{lstlisting}[mathescape=true, escapeinside=**]
| $\m{entry}$ {
    def $X$(*\m{this}*){
        spawn *\m{worker}* with
            *\m{this}*?*\m{client}*; *\m{client}*!$res$; *\m{client}*&{ $next$: $X$(*\m{worker}*), $end$: $\nil$ }
        continue *\m{worker}*$\tricom$*\m{client}*; $\nil$
    }
    main{ *\m{client}*?; $X$(*\m{entry}*) }
}
\end{lstlisting}
\end{minipage}

  After the initial communication and procedure call, variable \m{w} in \m{client} needs to be
  renamed to \m{worker} in order for the next communication to reduce.
\label{example:serverlessnetwork}
\end{example}

\section{Extraction with process spawning}
\label{section:processspawning}

In the presence of spawning, the intuitive process of extraction described earlier no longer works:
since a network can generate an unbounded number of new processes, there is no guarantee that the
SEG is finite and, as a consequence, that the procedure terminates.

However, we observe that since networks are finite and can only reduce to networks built from their subterms, there is only a finite number of possible behaviours for processes that are spawned at runtime.
Therefore we can keep SEGs finite if we allow renaming processes when connecting nodes. This intuition is key to our development.

\begin{example}
  Consider the network from \cref{example:serverlessnetwork} and its SEG, shown in
  \cref{fig:segserverless}.
  The dotted part shows the network that would be generated by symbolic execution as described
  earlier.
  By allowing renaming of processes, we can close a loop by applying the mapping
  $\{\m{client}\mapsto\m{client}, \m{entry/worker0}\mapsto\m{entry}\}$.
  The parameters \m{w} and \m{worker} are both variables mapping to \m{entry/worker0}, and since
  \m{entry} has terminated, the remapping makes the dotted node equivalent to the second node of the
  SEG, as shown by the loop.
  For simplicity we only show the process variables that are changed by the mapping, i.e., we omit
  $\m{client}\mapsto\m{client}$.
\end{example}
\begin{figure}[!htb]
\includegraphics[width=\textwidth]{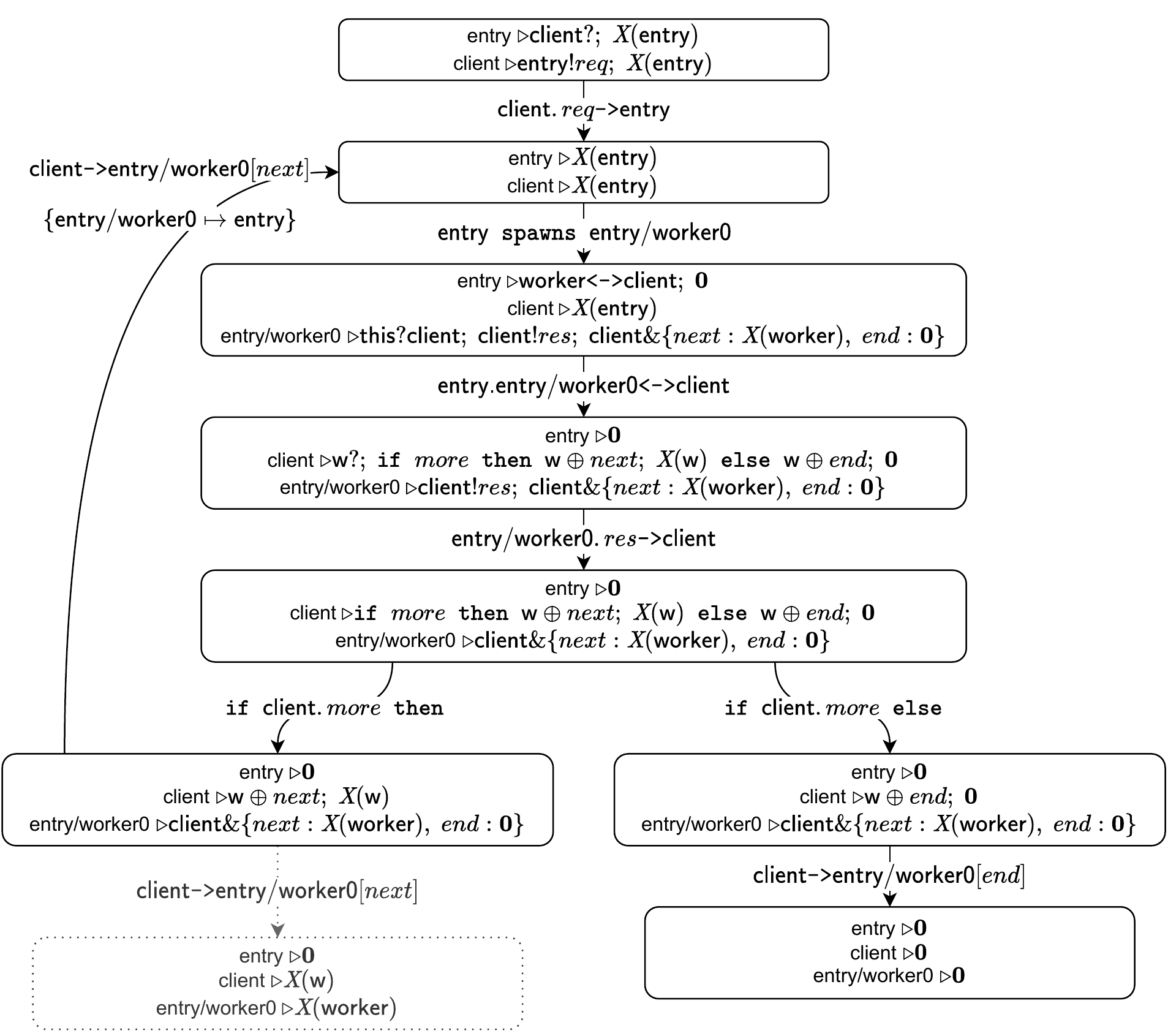}
\caption{The SEG of the network in \cref{example:serverlessnetwork}.}
\label{fig:segserverless}
\end{figure}

\subsection{Generating SEGs}
Formally, we define an abstract semantics for networks that makes two changes with respect to the
concrete semantics given above.
First, we remove all information about states $\sigma$ (and as a consequence about the actual values
being communicated), as in~\cite{CLM17}.
Secondly, we now treat \emph{all} process names as variables, and replace the communication graph
$\mathcal G$ by a partial function $\gamma$ mapping pairs of a process name and a process variable
to process names: intuitively, $\gamma(\m{p},\m{q})$ returns the name of the actual process that
\m{p} locally identifies as $\m{q}$.
If $\gamma(\m{p},\m{q})$ is undefined, then $\m{p}$ does not know who to communicate with.
We assume that initially $\gamma(\m{p},\m{p})=\m{p}$ for all \m{p}, and that, for all \m{p} and
\m{q}, either $\gamma(\m{p},\m{q})=\m{q}$ (meaning that $\m{p}$ knows $\m{q}$'s name and is allowed
to communicate with it) or $\gamma(\m{p},\m{q})$ is undefined (meaning that $\m{p}$ is not connected
to \m{q} and cannot communicate with it) -- this allows us to model different initial network
topologies.
\cref{fig:abstractnetworksemantics} shows the abstract versions of the rules previously shown.
\begin{figure}[!ht]
\[
\infer[\rname{N|Com}]
      {\m{p} \triangleright \m{r}!e; B_1 \mid \m{q} \triangleright \m{s}?; B_2,\gamma
        \reduce{\m{p}.e \com \m{q}}
        \m{p} \triangleright B_1 \mid \m{q} \triangleright B_2,\gamma}
      {\m{r} \downarrow^\gamma_\m{p} \m{q} \quad \m{s} \downarrow^\gamma_\m{q} \m{p}}
\]

\[
\infer[\rname{N|Intro}]
      {\begin{array}{c}
          \m{p} \triangleright \m{s}\tricom \m{t}; B_\m{p} \mid \m{q} \triangleright \m{u}?\m{w}; B_\m{q} \mid \m{r} \triangleright \m{v}?\m{x}; B_\m{r},\gamma \\ 
          \reduce{\m{p}.\m{q} \tricom \m{r}}\\
          \m{p} \triangleright B_\m{p} \mid \m{q} \triangleright B_\m{q} \mid \m{r} \triangleright B_\m{r}, \gamma[\langle \m{q,w} \rangle \mapsto \m{r}][\langle \m{r,x} \rangle \mapsto \m{u}]
      \end{array}}
      {\m{s} \downarrow^\gamma_\m{p} \m{q} \quad \m{t} \downarrow^\gamma_\m{p} \m{r} \quad \m{u} \downarrow^\gamma_\m{q} \m{p} \quad \m{v} \downarrow^\gamma_\m{r} \m{p}}
\]

\vspace{\baselineskip}
\[
\infer[\rname{N|Spawn}]
      {\begin{array}{c}
          \m{p} \triangleright \texttt{spawn \m{q} with $B_q$ continue $B$},\gamma \\
          \reduce{\m{p} \texttt{ spawns } \m{q}} \\
          \m{p} \triangleright B \mid \m{r} \triangleright B_q, \gamma[\langle \m{p,q} \rangle \mapsto \m{r}][\langle\m{r,p}\rangle\mapsto\m{p}]
      \end{array}}
      {} 
\]
\caption{Abstract semantics of networks, selected rules}
\label{fig:abstractnetworksemantics}
\end{figure}

\begin{example}
  It is easy to check that the SEG shown in \cref{fig:segserverless} follows the rules in
  \cref{fig:abstractnetworksemantics}.
  Initially, the variable mapping is the identity, and this remains unchanged after the first
  communication.

  \begin{center}\begin{tabular}{c|cccc}
    $\gamma$      &\ \m{entry}\ &\ \m{client}\ &\ \m{worker}\ &\ \m{w}\ \\ \hline
      \m{entry}   & \m{entry} & \m{client} & ---        & ---   \\
    \ \m{client}\ & \m{entry} & \m{client} & ---        & ---
  \end{tabular}\end{center}

  When \m{entry} spawns the new process \m{entry/worker0}, this name is associated to \m{entry}'s
  local variable \m{worker} according to rule \rname{N|Spawn}.
  So the variable mapping is now given by the following table.

  \begin{center}\begin{tabular}{c|cccc}
    $\gamma$          & \m{entry} & \m{client} & \m{worker}        &\ \m{w}\ \\ \hline
    \m{entry}         & \m{entry} &\ \m{client}\ &\ \m{entry/worker0}\ & ---   \\
    \m{client}        & \m{entry} & \m{client} & ---               & ---   \\
    \ \m{entry/worker0}\ &\ \m{entry}\ & ---        & ---               & ---
  \end{tabular}\end{center}

  Next, \m{entry} introduces \m{entry/worker0} and \m{client} to each other; \m{client} uses the
  local name \m{w} for the new process.
  According to rule \rname{N|Intro}, the variable mapping is now the following.

  \begin{center}\begin{tabular}{c|cccc}
    $\gamma$          & \m{entry} & \m{client} & \m{worker}        & \m{w} \\ \hline
    \m{entry}         & \m{entry} & \m{client} &\ \m{entry/worker0}\ & ---   \\
    \m{client}        & \m{entry} & \m{client} & ---               &\ \m{entry/worker0}\ \\
    \ \m{entry/worker0}\ &\ \m{entry}\ &\ \m{client}\ & ---               & ---
  \end{tabular}\end{center}
\end{example}

To determine whether we can close a loop in the SEG, we need the following definitions.
\begin{definition}
  Two networks $N$ and $N'$ are \emph{equivalent} if there exists a total bijective mapping $M$ from
  processes in $N$ to processes in $N'$, such that, for all processes \m{p}:
  \begin{itemize}
  \item if \m{p} has main behaviour $B$, then $M(\m{p})$ has main behaviour $M(B)$ (where $M$ is
    extended homeomorphically to behaviours);
  \item if \m{p} has not terminated and $X(\tilde{\m{q}})=B$ is a procedure definition in \m{p},
    then $X=M_{\m{q}}(B)$ is a procedure definition in $M(\m{p})$, where $M_{\m{q}}$ maps every
    process in $\tilde{\m{q}}$ to itself and every other process \m{r} to $M(\m{r})$.
  \end{itemize}
\end{definition}

To close a loop in the SEG, we also need to look at the variable mappings $\gamma$.

\begin{definition}
  Two nodes $N,\gamma$ and $N',\gamma'$ in a SEG are \emph{behaviourally equivalent} if:
  \begin{itemize}
  \item There exists a series of reductions $\tilde{\lambda}$ in the SEG such that
    $N,\gamma\reduce{\tilde{\lambda}}^* N',\gamma'$;
  \item There exists a mapping function $M$ that proves $N$ and $N'$ are equivalent.
  \item If $M(\m{p})=\m{q}$, $\gamma (\m{p,a})=\m{b}$,
    $\gamma'(\m{q,a})=\m{c}$, and there is a reduction accessible from
    $N,\gamma$ that evaluates $\gamma^*(\m{p,a})$ for some intermediary $\gamma^*$, then $M(\m{b})=\m{c}$.
  \end{itemize}
  \label{def:nodeequal}
\end{definition}

The last point of the definition only applies to variables actually evaluated in reductions of $N$.
This makes extraction more efficient (as more nodes are equivalent): a variable might have been used
for a previous step in the evolution up to $N,\sigma$, but if it remains unused thereafter it does
not affect the behaviour anymore, and can be ignored.

\begin{lemma}
  Let $N,\gamma$ and $N',\gamma'$ be behaviourally equivalent nodes in a SEG for a given network.
  The graph obtained by redirecting all edges coming into $N',\gamma'$ to $N,\gamma$ and removing
  the nodes that are no longer accessible from the root is also a SEG for the original network.
\label{lemma:behaviouralloops}
\end{lemma}

Our implementation simply looks for a suitable $N,\gamma$ when $N',\gamma'$ is generated, and
records the mapping $M$ in the edge leading to $N,\gamma$.

\subsubsection{Detecting unrestricted spawning.}
It is possible to create a network where processes are spawned in a loop, faster than they
terminate, making the number of processes increase for every iteration. 
Such networks embody a resource leak, and they cannot be extracted by our theory.
To ensure termination, our algorithm must be able to detect resource leaks, which is an undecidable
problem.
We deal with it as follows: when a new candidate node is generated, we check whether there is a
\emph{surjective} mapping with the properties described above such that there are at least two
process values mapped to the same process name.
If this happens, the algorithm returns failure.
We include some examples of networks with resource leaks in \cref{section:leakexamples}.

\subsection{Generating the choreography}
Building a choreography from a SEG is similar to the original case.
The main change deals with procedure calls: we use the variable mappings in edges that close loops
to determine their parameters and arguments.

When unrolling the SEG, each node corresponding to a procedure definition gets a list of parameters
corresponding to the process names\footnote{Assuming some predefined ordering of process names.} that
appear in the co-domain of any variable mapping in an edge leading to that node.
Each procedure call is then appended to the reverse images of these processes by the mapping in the
edge leading to it.
Note that the edges do not contain process names that are mapped to themselves; as such, processes
that are the same in all maps will not appear as arguments to the extracted procedure.
A consequence of this is that some procedures may get an empty set of arguments.

After this transformation the choreography can again be extracted by recursively traversing the
resulting forest.

We formulate the correctness of our extraction procedure in terms of strong bisimilarity~\cite{S11}.

\begin{theorem}
  \label{thm:bisim}
  If $C$ is a choreography extracted from a network $N$, then $C \sim N$.
\end{theorem}

\begin{example}
  We return to the network in \cref{example:serverlessnetwork}, whose SEG was shown in
  \cref{fig:segserverless}.
  The only loop node has two incoming edges, one with empty (identity) mapping and another with
  $\{\m{entry/worker0}\mapsto\m{entry}\}$.
  Therefore this node is extracted to a procedure $X1$ with one process variable \m{entry}.
  The original call simply instantiates this parameter as itself, while the recursive call replaces
  it with \m{entry/worker0}.
  
  The choreography extracted from this SEG is thus the following.
\begin{lstlisting}[mathescape=true, escapeinside=**]
def $X1$(*\m{entry}*) { 
    *\m{entry}* spawns *\m{entry/worker0}*; *\m{entry.entry/worker0}*$\tricom$*\m{client}*;
    *\m{entry/worker0}*.$res\com$*\m{client}*;
    if *\m{client}*.*\mi{more}*
        then *\m{client}*$\com$*\m{entry/worker0}*$[next]$; $X1$(*\m{entry/worker0}*) 
        else *\m{client}*$\com$*\m{entry/worker0}*$[end]$; $\nil$ 
} 
main { *\m{client}*.$req\com$*\m{entry}*; $X1$(*\m{entry}*) }
\end{lstlisting}
\end{example}

\subsection{Implementation and limitations}
\label{sec:limitations}
The extension of the original extraction algorithm to networks with process spawning has been
implemented in Java.
It can successfully extract the networks in the examples given here, as well as a number of randomly
generated tests following the ideas from~\cite{CLMS22}.
Due to space constraints, we do not report on the details of our testing strategy, which
is an extension of the strategy presented in detail in~\cite{CLMS22}, extended in the natural way to
include networks with process spawning and introduction.

Since our language only allows for tail recursion, divide-and-conquer algorithms such as mergesort are
currently still not extractable, and our next plan is to extend the algorithm to deal with general
recursion.
This is not a straightforward extension, as our way of constructing the SEG has no way of getting
past a potentially infinite recursive subterm to its continuation.\footnote{This was also the reason
  for only including tail recursion in the original work~\cite{CLM17}.}

Another example of an unextractable network, which does not use general recusion, is the following.
\begin{lstlisting}[mathescape=true, escapeinside=**]
*\m{s}* {
	def $X$(*\m{p,t}*){ 
		if $cont$ then spawn *\m{q}* with $X$(*\m{t,q}*) continue *\m{q}*?; *\m{p}*!$m$; $\nil$ else *\m{p}*!$m$; $\nil$
	} 
	main{ $X$(*\m{p,s}*) } 
} |
p{s?; stop}
\end{lstlisting}
Although the spawned processes behave as their parent, the entire network never repeats itself, and
extraction fails: extracting a choreography would require closing a loop where some processes did
not reduce.
This is essentially the same limitation already discussed in~\cite{CLMS22}, and cannot be avoided:
given that the problem of determining whether a network can be represented by a choreography is in
general undecidable~\cite{CLM17}, soundness of our algorithm implies that such networks will
always exist.

\section{Conclusion}
\label{sec:concl}
We showed how the state-of-the-art algorithm for choreography extraction~\cite{CLM17,CLMS22} could
be extended to accommodate for networks with process spawning.
This adaptation requires allowing processes names to change dynamically, so that the total number of
networks that needs to be consider remains finite.
The resulting theory captures examples including loops where processes that are spawned at runtime
take over for other processes that terminate in the meantime.
This extension also required adding parameterised procedures to the network and choreography
language, and including a form of resource leak detection to ensure termination.

A working implementation of choreography extraction with process spawning is available at \cite{implementation}.

\paragraph{Acknowledgements.} This work was partially supported by Villum Fonden, grant no.\ 29518.

\bibliography{sources}

\appendix

\section{Full definitions and proofs}
\label{section:fullproofs}
For completeness, we include the full definitions of network and choreography semantics, as well as
the proofs of our main results.

\subsection{Networks}
Behaviours are defined by the following grammar.
\begin{align*}
  B ::= & \ \nil
  \mid X
  \mid \m{p}!m; B
  \mid \m{p}?; B
  \mid \m{p}\oplus \ell; B
  \mid \m{p}\& \{\ell_1:B_1,\ldots,\ell_n:B_n\} \\
  \mid \m{q} \tricom \m{r}; B
  &
  \mid \m{p}?\m{t};\; B_{\m q}
  \mid \texttt{if }\! e \!\texttt{ then }\! B_1 \texttt{ else }\! B_2
  \mid \texttt{spawn } \m{q} \texttt{ with } B_{\m q} \texttt{ continue } B
\end{align*}
A process consists of a set of procedure definitions $\mathcal{D}_\m{p}$ together with a (main)
behaviour $B$, also written
$\m{p}\triangleright\{\texttt{def }X_i(\tilde{\m{p}})=B_i\}_{i\in I} \texttt{ in } B$.
A network is a finite set of processes running in parallel, together with a state $\sigma$ mapping
pairs of a process name and variable name to the process's variable's value, and an undirected graph
$\mathcal G$ whose nodes are the names of the processes in the network.

We use the following notations:
\begin{itemize}
\item $\sigma[\langle \m{p},x\rangle \mapsto v]$ denotes the update of $\sigma$ where \m{p}'s
  variable $x$ now maps to value $v$;
\item $e \downarrow^\sigma_\m{p} v$ denotes evaluating expression $e$ at process \m{p} under state
  $\sigma$ reduces to value $v$;
\item $\m{p}\leftrightarrow\m{q}\in\mathcal G$ denotes that there is an edge between \m{p} and \m{q}
  in $\mathcal G$;
\item $\mathcal G\cup\{\m{p}\leftrightarrow\m{q}\}$ denotes the graph obtained from $\mathcal G$ by
  adding an edge between \m{p} and \m{q}.
\end{itemize}

The semantics of networks is given in form of transitions
$N,\sigma,\mathcal G \reduce{\lambda} N',\sigma',\mathcal G'$, where the \emph{transition label}
$\lambda$ describes the action being executed in the transition.
The rules for reductions are those of \cref{fig:networksemantics}.
The set $\mathcal D$ is left implicit in these rules, as it never changes.
\begin{figure}[!htb]
\[
\infer[\rname{N|Com}]
      {\m{p} \triangleright \m{q}!e; B_1 \mid \m{q} \triangleright \m{p}?; B_2,\sigma,\mathcal G
        \reduce{\m{p}.v \com \m{q}}
        \m{p} \triangleright B_1 \mid \m{q} \triangleright B_2,\sigma[\langle\m{q},x \rangle \mapsto v],\mathcal G}
      {\m{p}\leftrightarrow\m{q}\in\mathcal G & e \downarrow^\sigma_\m{p} v}
\]

\[
\infer[\rname{N|Sel}]
      {\m{p} \triangleright \m{q} \oplus \ell_j; B \mid \m{q} \triangleright \m{s}\&\{\ell_1:B_1,\ldots,\ell_n:B_n\},\sigma,\mathcal G
        \reduce{\m{p} \com \m{q}[\ell_j]}
        \m{p} \triangleright B \mid \m{q} \triangleright B_j,\sigma,\mathcal G}
      {\m{p}\leftrightarrow\m{q}\in\mathcal G & 1 \leq j \leq n}
\]

\[
\infer[\rname{N|Intro}]
      {\begin{array}{c}
          \m{p} \triangleright \m{q}\tricom \m{r}; B_\m{p} \mid
          \m{q} \triangleright \m{p}?\m{r}; B_\m{q} \mid
          \m{r} \triangleright \m{p}?\m{q}; B_\m{r}, \sigma, \mathcal G\\
          \reduce{\m{p}.\m{q} \tricom \m{r}}\\
          \m{p} \triangleright B_\m{p} \mid
          \m{q} \triangleright B_\m{q} \mid
          \m{r} \triangleright B_\m{r}, \sigma, \mathcal G\cup\{\m{q}\leftrightarrow\m{r}\}
      \end{array}}
      {\m{p}\leftrightarrow\m{q}\in\mathcal G & \m{p}\leftrightarrow\m{r}\in\mathcal G}
\]

\vspace{\baselineskip}
\[
\infer[\rname{N|Spawn}]
      {\begin{array}{c}
          \m{p} \triangleright \texttt{spawn \m{q} with $B_q$ continue $B$},\sigma,\mathcal G
          \reduce{\m{p} \texttt{ spawns } \m{q}}
          \m{p} \triangleright B \mid
          \m{q} \triangleright B_q, \sigma, \mathcal G\cup\{\m{p}\leftrightarrow\m{q} \}
      \end{array}}
      {}
\]

\vspace{\baselineskip}
\[
\infer[\rname{N|Then}]
      {\m{p} \triangleright \texttt{if $e$ then } B_1 \texttt{ else } B_2, \sigma,\mathcal G
        \reduce{\tau_\m{p}}
        \m{p}\triangleright B_1,\sigma,\mathcal G}
      {e \downarrow^\sigma_\m{p} \texttt{true}}
\]

\[
\infer[\rname{N|Else}]
      {\m{p} \triangleright \texttt{if $e$ then } B_1 \texttt{ else } B_2, \sigma,\mathcal G
        \reduce{\tau_\m{p}}
        \m{p}\triangleright B_2,\sigma,\mathcal G}
      {e \downarrow^\sigma_\m{p} \texttt{false}}
\]

\[
\infer[\rname{N|Par}]
      {N|M,\sigma,\mathcal G \reduce{\lambda} N'|M,\sigma',\mathcal G'}
      {N,\sigma,\mathcal G \reduce{\lambda} N',\sigma',\mathcal G'}
\qquad
\infer[\rname{N|Struct}]
      {N,\sigma,\mathcal G \reduce{\lambda} N',\sigma',\mathcal G'}
      {N \preceq M & M,\sigma,\mathcal G \reduce{\lambda} M',\sigma',\mathcal G' & M' \preceq N'}
\]

\caption{Semantics of networks}
\label{fig:networksemantics}
\end{figure}

Rule~\rname{N|Struct} closes reductions under a \emph{structural precongruence}, which allows
procedure calls to be unfolded and their parameters instantiated.
The key rule defining this relation is
\[
\infer[\rname{N|Unfold}]
      {\m{p} \triangleright X(\m{q,r,}\ldots)
        \preceq
        \m{p} \triangleright B_X[\m{q\leftarrow a},\m{r\leftarrow b},\ldots]}
      {X(\m{a,b,}\ldots) = B_X \in \mathcal{D}_\m{p}}
\]
and it is the only rule depending on the set $\mathcal D$.
Structural precongruence is closed under reflexivity, transitivity, and context.

For the extraction algorithm, we also need an abstract semantics.
The rules for this semantics obtained from those in \cref{fig:networksemantics} by (i)~removing the
state $\sigma$, (ii)~replacing the connection graph $\mathcal G$ with a variable mapping
$\gamma$ and (iii)~replacing the transition labels by the corresponding abstract ones (using
expressions instead of values in \rname{N|Com}, and using the more expressive labels
$\m{p}.e\ \m{then}$ ad $\m{p}.e\ \m{else}$ instead of $\tau_\m{p}$ in the rules for conditionals).
The only rules updating $\gamma$ are \rname{N|Spawn} and \rname{N|Intro}, which are given in the
main text.

The abstract semantics for networks generalises the concrete semantics, in the sense that if
$N,\sigma,\mathcal G\reduce\lambda N',\sigma',\mathcal G'$, then
$N,\gamma\reduce\lambda' N',\gamma'$ where $\lambda'$ is the abstract label corresponding to
$\lambda$ and $\gamma,\gamma'$ are variable mappings corresponding to $\mathcal G$ and $\mathcal
G$', respectively, in the sense explained in the main text.

\subsection{Choreographies}
Choreography bodies are generated by the grammar
\begin{align*}
  C ::= & \ \nil
  \mid X
  \mid \m{p}.m \com \m{q}; C
  \mid \m{p} \com \m{q}[\ell]; C
  \mid \texttt{if } \m{p}.e \texttt{ then } C_1 \texttt{ else } C_2 \\
  & \mid \texttt{\m{p} spawns \m{q}};\; C
  \mid \m{p.q} \tricom \m{r};\; C
\end{align*}

A choreography consists of a set of procedure definitions and a main choreography body, written
$\texttt{def}\{X_i(\tilde{\m{p}})=C_i\}_{i\in I} \texttt{ in }C$.
The semantics of choreographies is again given as transitions
$C,\sigma,\mathcal G \reduce{\lambda}C',\sigma',\mathcal G'$, with $C$ representing the main
choreography body, and is defined by the rules in \cref{fig:chorsemantics}.

\begin{figure}
\[
\infer[\rname{C|Com}]
      {\m{p}.e \com \m{q} ; C, \sigma,\mathcal G
        \reduce{\m{p}.v \com \m{q}}
        C,\sigma[\langle \m{q},x \rangle \mapsto v],\mathcal G}
      {\m{p}\leftrightarrow\m{q}\in\mathcal G & e \downarrow^\sigma_\m{p} v}
\quad\quad
\infer[\rname{C|Sel}]
      {\m{p} \com \m{q}[\ell] ; C, \sigma,\mathcal G
        \reduce{\m{p} \com \m{q}[\ell]}
        C,\sigma,\mathcal G}
      {\m{p}\leftrightarrow\m{q}\in\mathcal G}
\]

\[
\infer[\rname{C|Intro}]
      {\m{p}.\m{q} \tricom \m{r} ; C, \sigma,\mathcal G
        \reduce{\m{p}.\m{q} \tricom \m{r}}
        C,\sigma,\mathcal G\cup\{\m{q}\leftrightarrow\m{r}\}}
      {\m{p}\leftrightarrow\m{q}\in\mathcal G & \m{p}\leftrightarrow\m{r}\in\mathcal G}
\]

\[
\infer[\rname{C|Then}]
      {\texttt{if \m{p}.$e$ then } C_1 \texttt{ else } C_2, \sigma,\mathcal G
        \reduce{\tau_\m{p}}
        C_1,\sigma,\mathcal G}
      {e \downarrow^\sigma_\m{p} \texttt{ true}}
\]

\[
\infer[\rname{C|Else}]
      {\texttt{if \m{p}.$e$ then } C_1 \texttt{ else } C_2, \sigma,\mathcal G
        \reduce{\tau_\m{p}}
        C_2,\sigma,\mathcal G}
      {e \downarrow^\sigma_\m{p} \texttt{ false}}
\]

\[
\infer[\rname{C|Spawn}]
      {\m{p} \texttt{ spawns } \m{q} ; C, \sigma,\mathcal G
        \reduce{\m{p} \texttt{ spawns } \m{q}}
        C,\sigma,\mathcal G\cup\{\m{p}\leftrightarrow\m{q}\}}
{}
\]

\[
\infer[\rname{C|Struct}]
      {C_1,\sigma,\mathcal G \reduce{\lambda} C_1',\sigma',\mathcal G'}
      {C_1 \preceq C_2 & C_2,\sigma,\mathcal G \reduce{\lambda}C_2',\sigma',\mathcal G' & C_2' \preceq C_1'}
\]
\caption{Semantics of choreographies}
\label{fig:chorsemantics}
\end{figure}

Rule \rname{C|Struct} again closes reductions under a structural precongurence $\preceq$, which not
only allows procedure calls to be unfolded as before, but also allows non-interfering operations to
execute in any order.
The key rules for structural precongruence in choreographies are defined in \cref{fig:chorprecon}.

\begin{figure}
\[
\infer[\rname{C|Eta-Eta}]
      {\eta;\eta';C \preceq \eta';\eta;C}
      {\m{pn(\eta)} \cap \m{pn(\eta')} = \emptyset}
\]

\[
\infer[\rname{C|Eta-Cond}]
{\texttt{if \m{p}.$e$ then } (\eta; C_1) \texttt{ else } (\eta; C_2) \preceq \eta; \texttt{if \m{p}.$e$ then } C_1 \texttt{ else } C_2}
{\m{p} \notin \m{pn(\eta)}}
\]

\[
\infer[\rname{C|Cond-Eta}]
{\eta; \texttt{if \m{p}.$e$ then } C_1 \texttt{ else } C_2 \preceq \texttt{if \m{p}.$e$ then } (\eta; C_1) \texttt{ else } (\eta; C_2)}
{\m{p} \notin \m{pn(\eta)}}
\]

\[%\texttt{} spacing replaced by \; space to avoid horizontal overflow
\infer[\rname{C|Cond-Cond}]
{\begin{array}{c}
\begin{array}{rl}
\texttt{if \m{p}.$e$ then}\; & (\texttt{if \m{q}.$e'\;$then}\; C_1\; \texttt{else}\; C_2)\;\\ 
\texttt{else}\; & (\texttt{if \m{q}.$e'\;$then}\; C_1'\; \texttt{else}\; C_2')
\end{array}\\
\preceq\\
\begin{array}{rl}
\texttt{if \m{q}.$e'\;$then}\; & (\texttt{if \m{p}.$e$ then}\; C_1\; \texttt{else}\; C_1')\;\\
\texttt{else}\; & (\texttt{if \m{p}.$e$ then}\; C_2\; \texttt{else}\; C_2')
\end{array}
\end{array}}
{\m{p}\neq\m{q}}
\]

\[
\infer[\rname{C|Unfold}]
{X(\m{a,b,}\ldots) \preceq C_X[\m{p}\leftarrow\m{a}, \m{q}\leftarrow\m{b},\ldots]}
{X(\m{p,q,}\ldots) = C_X \in \mathcal{D}}
\]
\caption{Structural precongruence in choreographies}
\label{fig:chorprecon}
\end{figure}

Rule \rname{Eta-Eta} swaps non-interfering adjacent communications (these can be value
communication, label selection, introduction actions, or spawning of new processes).
Function \m{pn()} returns the set of process names involved in $\eta$, so the premise of the rule is that the two adjacent communications do not have any processes names in common.
Likewise, rules \rname{C|Eta-Cond} and \rname{C|Cond-Eta} allow for swapping conditionals with
interactions, and \rname{C|Cond-Cond} for swapping conditionals at different processes.

\subsection{Extraction}
\label{sec:procedureextract}
For completeness, we recap some of the definitions from~\cite{CLM17} that are unchanged (although
they now have a wider scope) in this work.
Throughout most of this section we work only with the abstract semantics for networks.

\begin{definition}
The \emph{Abstract Execution Space (AES)} of a network $N$ and variable mapping $\gamma$ is a
directed graph whose nodes are all pairs $N',\gamma'$ such that
$N,\gamma\reduce{\lambda_1}\cdots\reduce{\lambda_n} N',\gamma'$, and such that there is an edge from
$N_1,\gamma_1$ to $N_2,\gamma_2$ with label $\lambda$ iff $N_1,\gamma_1\reduce{\lambda}N_2,\gamma_2$.
\label{def:aes}
\end{definition}
The AES is an abstract representation of all possible executions of $N$.

\begin{definition}
A \emph{Symbolic Execution Graph (SEG)} for $N,\gamma$ is a subgraph of the AES for $N,\gamma$ that
contains $N,\gamma$, and such that every node $N',\gamma'$ with $N'\neq\nil$ has either one outgoing
edge labelled by an interaction, or two outgoing edges labelled $\m{p}.e \textup{\texttt{ then}}$
and $\m{p}.e \textup{\texttt{ else}}$ respectively.
\label{def:seg}
\end{definition}
An SEG fixes an order of execution of (inter)actions, representing just a single evolution of the
network.
Confluence of the network semantics implies that the existence of a SEG is independent of this
order.

Our extraction algorithm relies on building a (finite) SEG in finite time.
The only aspects not discussed in the main text regard restricting (i)~the usage of rule
\rname{N|Struct} (for termination) and (ii)~the closure of loops (for soundness).

\subsubsection{Unfolding procedures.}
We restrict the AES to transitions that only apply \rname{N|Unfold} to processes of the form
$\m{p}\triangleright X$, where $\m{p}$ also appears in the label of the reduction.
In the setting of~\cite{CLM17} this guarantees that the AES is finite.

\subsubsection{Restricting loop closure.}
Soundness of extraction requires that any sequence of actions executable from $N,\sigma,\mathcal G$
be executable by $C,\sigma,\mathcal G$ if $C$ is obtained by extraction from $N$.
Our definition of SEG allows for loops where some processes in the network never reduce
(see~\cite{CLM17} for examples), from which we would extract unsound choreographies.
To avoid this, we annotate all process in networks with a (Boolean) marking, which is also checked
when comparing nodes.
All processes are initially unmarked; when there is an edge $N,\gamma\reduce\lambda N',\gamma'$, the
processes appearing in $\lambda$ become marked.
If this results in all processes in $N'$ being marked, all markings are erased instead.
A \emph{valid} SEG is one where every loop contains a node where all processes are unmarked; we only
allow extraction from valid SEGs.

The implementation of this check is computationally expensive, so instead we use a counter that
counts how many times the marking has been reset in the branch leading from the root to the current
node~\cite{CLMS22}.
This requires that the two branches of a conditional are generated independently and not allowed to
have edges between them, which we address by including an additional ``choice path'' to every node
and allowing loops to be closed only when the target node's choice path is a subsequence of the
origin's.

These considerations are orthogonal to the current work, and they are unchanged in the current
implementation.

The new ingredient in loop closure is that we allow variable renaming.
Soundness of this construction is stated in \cref{lemma:behaviouralloops}.
\begin{proof}[\cref{lemma:behaviouralloops}]
Let $\tilde{\lambda}$ and $M$ be as in \cref{def:nodeequal}. Then let $\tilde{\lambda'}$ be the same series of interactions as in $\tilde{\lambda}$, but where every process name is replaced with the process they map to in $M$. It is possible to apply this new series of interactions by $N',\gamma'\reduce{\tilde{\lambda'}}^*N'',\gamma''$ since by definition, the processes in $N$ can be renamed to obtain $N'$, and for every renaming using $M(\m{p})=\m{q}$, all variables $\langle \m{p,a}\rangle\mapsto \m{b}\in\gamma$  and $\langle \m{q,a}\rangle\mapsto \m{c}$ implies the renaming $M(\m{b})=\m{c}$ too. Since the same renaming of processes is done in $\tilde{\lambda'}$, then $N,\gamma\reduce{\tilde{\lambda}}^* N',\gamma'$ implies $N',\gamma'\reduce{\tilde{\lambda'}}^*N'',\gamma''$. Furthermore, $N',\gamma$ and $N'',\gamma''$ are also behaviourally equivalent, as they follow the same reductions that previously lead to a behaviourally equivalent network, so the entire process can be repeated, by applying the renaming using the new mapping $M'$ to $\tilde{\lambda'}$ to obtain $\tilde{\lambda''}$, which can then be used to reduce further by $N'',\gamma''\reduce{\tilde{\lambda''}}^*N''',\gamma'''$, and so on.

Since the reductions show recursive behaviour, creating a loop after applying every reduction in $\tilde{\lambda}$ to $N$ does indeed capture the recursive behaviour of the network.
\qed
\end{proof}

\subsection{Soundness and completeness}
\label{section:bisimilarity}
We now show that choreographies extracted from valid SEGs are bisimilar to the network they are
extracted from.
The proof is very similar to the one in~\cite{CLMS22}, with only minor modifications to account for
the fact that we now use process variables in the SEG.

Throughout this section, let $N$ be a network and $\langle\mathcal D,C\rangle$ be the choreography
extracted from $N$ for a particular SEG $S$ (where $C$ is the main choreography body and
$\mathcal D$ is the set of procedure definitions).
Also let $\sigma$ and $\mathcal G$ be a state and a connection graph, intuitively the ``initial''
state and connection graph of $N$.

Define sequences $\{\lambda_i\}_{i\in I}$, $\{C_i\}_{i\in I}$, $\{\sigma_i\}_{i\in I}$, and $\{\mathcal G_i \}_{i \in I}$ of possibly infinite length as:
\begin{itemize}
\item $C_0=C$, $\sigma_0=\sigma$ and $\mathcal G_0 = \mathcal G$.
\item For each $i$, $\lambda_i$ is the label of the transition executing the head action in $C_i$,
  that is, the only action that can be executed without applying any of the structural precongruence
  rules other than \rname{C|Unfold}.
\item For each $i$, $C_{i+1},\sigma_{i+1},\mathcal G_{i+1}$ are the only choreography, state and
  connection graph such that
  $C_i,\sigma_i,\mathcal G_{i} \reduce{\lambda_i} C_{i+1},\sigma_{i+1},\mathcal G_{i+1}$.
  (Note that these are uniquely defined from the transition label.)
\item $I=\{0,\ldots,n\}$ if $C_n$ is $\nil$ for some $n$, and $\mathbb{N}$ otherwise.
\end{itemize}
Intuitively, these sequences represent one execution path for $C$ -- namely, the path determined by
$S$.

The first step is establishing that this execution path can be mimicked by $N$.
\begin{lemma}
  There exists a sequence of networks $\{N_i\}_{i\in I}$ such that $N_0=N$, and
  $N_i,\sigma_i,\mathcal G_i\reduce{\lambda_i}N_{i+1},\sigma_{i+1},\mathcal G_{i+1}$.
\label{lemma:networksequence}
\end{lemma}
\begin{proof}
  We prove by induction on $i$ that there exists a sequence of nodes $\{n_i\}_{i\in I}\subseteq S$
  such that: $N_i$ is obtained from $n_i$ by applying the composition of all variable mappings in
  the path from $n_0$ to $n_i$ to the network in $n_i$.

  For $i=0$ the thesis trivially holds (with empty path and identity variable mapping).
  Assume by the induction hypothesis that it holds for $i$.
  Due to how the sequence $\{\lambda_i\}_{i\in I}$ was constructed, there must be an outgoing edge
  of $n_i$ whose label $\lambda'_i$ is the abstraction of an action executable from
  $N_i,\sigma_i,\mathcal G_i$, modulo the current assignment of process names to process variables.
  The target node $n_{i+1}$ of this edge satisfies the thesis by construction.
  \qed
\end{proof}

From this point onwards we denote by $M_k$ be the composition of all variable mappings in the path
from $n_0$ to $n_k$.

\begin{lemma}
  For every $i\in I$ and transition label $\lambda$, it is the case that $C_i,\sigma_i,\mathcal G_i$
  can execute a reduction labelled by $\lambda$ iff $N_i,\sigma_i,\mathcal G_i$ can execute a
  reduction labelled by $\lambda$.
\label{lemma:chorreduceiffnetworkreduce}
\end{lemma}
\begin{proof}
  Assume that $C_i,\sigma_i,\mathcal G_i\reduce{\lambda}C',\sigma',\mathcal G'$ for some $C'$,
  $\sigma'$ and $\mathcal G'$.
  Let $j\geq i$ be the minimal index such that $\lambda$ and $\lambda_jM_j$ share a process name.
  Since structural precongruence can only exchange actions that do not share process names, and the
  semantics of choreographies only allows a process one action at each point in the sequence, it
  follows that $\lambda=\lambda_jM_j$.
  Furthermore, since no actions in $\lambda_iM_i,\ldots,\lambda_{j-1}M_{j-1}$ share process names
  with $\lambda$, if \m{p} is a process appearing in $\lambda$ it follows that (i)~\m{p}'s behaviour
  in $N_i,\ldots,N_j$ does not change, (ii)~$\sigma_i(\m{p})=\sigma_j(\m{p})$, and (iii)~the
  edges between \m{p} and other processes in $\lambda$ are the same in $\mathcal G_i$ and
  $\mathcal G_j$.

  Since $N_j,\sigma_j$ can execute $\lambda$ and the conditions for executing an action only depend
  on the processes involved in the action (their state, behaviour and interconnections), this
  implies that $N_i,\sigma_i,\mathcal G_i$ can reduce with label $\lambda$ to some network $N'$.
  An analysis of the rules for both choreography and network semantics shows that necessarily the
  resulting state and connection graph must be precisely $\sigma'$ and $\mathcal G'$, respectively.
  
  Conversely, assume that $N_i,\sigma_i,\mathcal G_i\reduce{\lambda}N',\sigma',\mathcal G'$ for some
  $N'$, $\sigma'$ and $\mathcal G'$.
  Since the processes involved in $\lambda$ cannot participate in any other reductions, the
  reduction labelled by $\lambda$ must remain enabled until it is executed.
  Furthermore, the restrictions on valid SEGs imply that it is executed, so there must exist
  $j\geq i$ such that $\lambda=\lambda_jM_j$ and $\lambda$ shares no process names with
  $\lambda_iM_i,\ldots,\lambda_{j-1}M_{j-1}$.
  Given how $\{\lambda_i\}_{i\in I}$ was defined from $C$, this implies that we can use structural
  precongruence to rewrite $C_i$ to allow executing $\lambda_j$.
  As before, the resulting state and connection graph must coincide with $\sigma'$ and
  $\mathcal G'$, respectively.
  \qed
\end{proof}

We now show that this particular reduction path captures all possible reduction paths.

\begin{lemma}
  Let $\tilde{\lambda'}=\lambda'_0,\ldots,\lambda'_j$ be a finite sequence of reductions labels such
  that $C,\sigma,\mathcal G\reduce{\tilde{\lambda'}}^* C',\sigma',\mathcal G'$.
  Then there exist $n\geq j$ and a permutation $\pi:\{0,\ldots,n\} \rightarrow \{0,\ldots,n\}$ such
  that $\lambda'_i=\lambda_{\pi(i)}\theta$ for $i=0,\ldots,j$, where $\theta$ is a renaming of
  spawned processes.\footnote{In other words, a mapping from process names to process names that
    is the identity for all processes initially in $C$.}
  Furthermore $\tilde{\lambda'}$ can be obtained from $\lambda_0,\ldots,\lambda_n$ by repeatedly
  transposing consecutive pairs of labels that do not share process names.
  \label{lemma:chorpermutation}
\end{lemma}
\begin{proof}
  We observe that choreography reductions are confluent up to renaming of spawned
  process.\footnote{The proof of this is a standard double induction on possible reductions.}

  The lemma follows by noting that it is always possible to choose a large enough $n$ such that all
  actions in $\tilde{\lambda'}$ occur in $\lambda_0,\ldots,\lambda_n$, modulo renaming of spawned
  processes.
  By confluence, it is possible to swap pairs of consecutive independent actions in
  $\lambda_0,\ldots,\lambda_n$ repeatedly until the sequence starts with
  $\lambda'_0,\ldots,\lambda'_j$.
  The composition of the indices of the actions being swapped yields the permutation $\pi$.
  \qed
\end{proof}

A similar argument establishes the next lemma.
\begin{lemma}
  Let $\tilde{\lambda'}=\lambda'_0,\ldots,\lambda'_j$ be a finite sequence of reductions labels such
  that $N,\sigma,\mathcal G\reduce{\tilde{\lambda'}}^* N',\sigma',\mathcal G'$.
  Then there exist $n\geq j$ and a permutation $\pi:\{0,\ldots,n\} \rightarrow \{0,\ldots,n\}$ such
  that $\lambda'_i=\lambda_{\pi(i)}\theta$ for $i=0,\ldots,j$, where $\theta$ is a renaming of
  spawned processes.
  Furthermore $\tilde{\lambda'}$ can be obtained from $\lambda_0,\ldots,\lambda_n$ by repeatedly
  transposing consecutive pairs of labels that do not share process names.
  \label{lemma:networkpermutation}
\end{lemma}

\begin{lemma}
  Let $\tilde{\lambda'}=\lambda'_0,\ldots,\lambda'_j$, where $\{\lambda'_i\}_{i\in I}$ is obtained
  from $\{\lambda_i\}_{i\in I}$ by repeatedly swapping consecutive elements of that share no process
  names.
  Then there exist a choreography $C'$, a network $N'$, a state $\sigma'$ and a connection graph
  $\mathcal G'$ such that:
  \begin{itemize}
  \item $C,\sigma,\mathcal G\reduce{\tilde{\lambda'}}^*C',\sigma',\mathcal G'$;
  \item $N,\sigma,\mathcal G\reduce{\tilde{\lambda'}}^*N',\sigma',\mathcal G'$;
  \item the actions that $C'$ and $N'$ can execute coincide.
  \end{itemize}
  \label{lemma:applicablecoincides}
\end{lemma}
\begin{proof}
  By induction on the number of elements swapped in constructing $\tilde{\lambda'}$.
  If that number is 0, then this is simply \cref{lemma:chorreduceiffnetworkreduce}.
  
  Assume by induction hypothesis that the thesis holds for $n$ swaps, and suppose that we now swap
  two consecutive labels $\lambda'_i$ and $\lambda'_{i+1}$ that share no process names.
  The thesis then follows from confluence of the semantics of both choreographies and networks.
  \qed
\end{proof}

We can now prove that $N$ and $C$ are bisimilar.
\begin{proof}[\cref{thm:bisim}]
  Define a relation $\mathcal R\subseteq\mathcal C\times\mathcal N$, where
  $$\mathcal{C}=\{C'\mid C,\sigma,\mathcal G\reduce{}^* C',\sigma',\mathcal G' \mbox{ for some $\sigma'$ and $\mathcal G'$}\}$$
  and
  $$\mathcal{N}=\{N'\mid N,\sigma,\mathcal G\reduce{}^* N',\sigma',\mathcal G'\mbox{ for some $\sigma'$ and $\mathcal G'$}\}$$
  by
  \[C'\mathcal R N' \mbox{ if there exists $\tilde{\lambda'}$ s.t. }
  C,\sigma,\mathcal G\reduce{\tilde{\lambda'}}^*C',\sigma',\mathcal G'
  \mbox{ and }
  N,\sigma,\mathcal G\reduce{\tilde{\lambda'}}^*N',\sigma',\mathcal G'\]
  
  Assume that $C'\mathcal R N'$.
  Then there exists a sequence of actions $\tilde{\lambda'}$ such that
  $C,\sigma,\mathcal G\reduce{\tilde{\lambda'}}C',\sigma',\mathcal G'$ and
  $N,\sigma,\mathcal G\reduce{\tilde{\lambda'}}N',\sigma',\mathcal G$.
  By \cref{lemma:chorpermutation}, $\tilde{\lambda'}$ can be obtained from $\tilde{\lambda}$ by
  repeatedly swapping adjacent, independent actions and renaming spawned processes.
  By \cref{lemma:applicablecoincides}, the actions that $C'$ and $N'$ can execute are the same.
  For each such action $\alpha$, \cref{lemma:chorpermutation} and \cref{lemma:applicablecoincides}
  can be applied to the sequence $\tilde{\lambda'};\alpha$ to conclude that if
  $C',\sigma',\mathcal G'\reduce{\alpha}C'',\sigma'',\mathcal G''$, then there exists an $N''$ such
  that $N',\sigma',\mathcal G'\reduce{\alpha}N'',\sigma'',\mathcal G''$ -- note that the renamings
  of spawned processes in both reduction sequences must coincide, so the end result is exactly the
  same.
  Conversely, if $N',\sigma',\mathcal G'\reduce{\alpha}N'',\sigma'',\mathcal G''$, then
  \cref{lemma:networkpermutation} and \cref{lemma:applicablecoincides} establish a similar
  correspondence.
  It follows that $\mathcal R$ is a bisimulation.
  \qed
\end{proof}

\subsection{Resource leak examples}
\label{section:leakexamples}
This sections contains examples of networks that cannot be extracted because they contain resource
leaks.
The networks are viable in the sense all communications match up, so they could be implemented as real systems.

\begin{example}
  We show a minimal resource leak example.
  The initial process is in an infinite loop where it repeatedly clones itself.
\begin{lstlisting}[mathescape=true, escapeinside=**]
*\m{p}*{ 
	def $X()${ spawn *\m{q}* with $X()$ continue $X()$ } 
	main{ $X()$ }
 }
\end{lstlisting}
\end{example}

\begin{example}
The initial process spawns several processes which after an setup phase will only interact among themselves and never terminate. The initial process then repeats the process indefinitely, adding more and more processes to the network.
\begin{lstlisting}[mathescape=true, escapeinside=**]
*\m{p}*{ 
	def $L1(\m{q})${ *\m{q}*!$data$; *\m{q}*?; $L1(\m{q})$ }
	def $L2(\m{a,b})${ *\m{a}*?; *\m{b}*?; *\m{b}*!$ok$; *\m{a}*!$ok$; $L2(\m{a,b})$ }
	def $X()${ 
		spawn *\m{a}* with *\m{p?q}*; $L1(\m{q})$ continue 
		spawn *\m{b}* with *\m{p?q}*; $L1(\m{q})$ continue 
		spawn *\m{c}* with *\m{p?a}*; *\m{p?b}*; $L2(\m{a,b})$ continue 
		*\m{a}\tricom\m{c}*; *\m{b}\tricom\m{c}*; $X()$ }
	main{ $X()$ }
 }
\end{lstlisting}
\end{example}

\end{document}